\documentclass[reprint,nofootinbib]{revtex4-2}
\usepackage{amsthm}
\theoremstyle{definition}
\usepackage{empheq}
\newtheorem*{theorem}{Theorem}
\newtheorem{corollary}{Corollary}

\usepackage{amsmath}
\usepackage{amssymb}
\usepackage{amsfonts}
\usepackage{bm}
\usepackage{microtype}
\usepackage{graphicx}
\usepackage{hyperref}

\DeclareMathOperator{\sgn}{sgn}
\renewcommand{\Im}{\operatorname{Im}}
\renewcommand{\Re}{\operatorname{Re}}
\let\originalleft\left 
\let\originalright\right
\renewcommand{\left}{\mathopen{}\mathclose\bgroup\originalleft}
\renewcommand{\right}{\aftergroup\egroup\originalright}
\newcommand{\del}{\bm{\nabla}}
\newcommand{\cross}{\times}
\newcommand{\nn}{\nonumber}
\newcommand{\beq}{\begin{equation}}
\newcommand{\eeq}{\end{equation}}

\begin{document}
\title{Exact solution to Maxwell's equations for the infinite ideal solenoid with a time-dependent surface current}
\author{Edward Parker}
\email{tparker@alumni.physics.ucsb.edu}
\date{\today}

\begin{abstract}
Very little previous literature has considered the \emph{exact} solution to Maxwell's equations for an infinite ideal cylindrical solenoid with an arbitrary time-dependent azimuthal surface current $K(t) \hat{\bm{\phi}}$. Most of the previous literature has focused on special cases and has approached the problem by calculating the magnetic vector potential $\bm{A}$, which requires performing some very complicated surface integrals over the cylinder. In this article, we take a simpler approach and directly tackle Maxwell's equations without ever invoking a vector potential. The high symmetry of the geometry allows us to reduce Maxwell's equations to just two coupled partial differential equations for two functions of two real variables, which can be readily solved numerically. We find the general analytic solution to these PDEs and derive the Green's functions for the electromagnetic fields, which allow us to calculate the fields directly from the surface current $K(t)$. We also briefly discuss a family of exact formal solutions that (the author believes) has not appeared in the previous literature because it corresponds to a current $K(t)$ that does not have a Fourier transform.
\end{abstract}

\maketitle

\section{Introduction}

Two of the workhorse examples in the study of classical electromagnetism (EM) are the parallel-plate capacitor and the cylindrical solenoid of azimuthal surface current. In the static context, these produce the simplest possible nontrivial electromagnetic fields far away from their edges: a uniform electric or magnetic field, respectively. Their highly symmetric geometries allow many quantities to be calculated exactly. In the dynamic context, the solenoid is arguably the ``nicer'' of these two examples, because we can consider a reasonably realistic idealization in which the surface current changes over time. By contrast, the continuity equation prevents us from moving charge between the plates of a parallel-plate capacitor without introducing a current that typically reduces the geometric symmetry.

Moreover, the solenoid with a slowly time-varying surface current is an excellent illustration of Faraday's law of induction that is discussed in most EM textbooks \cite{Griffiths, Jackson}. The high symmetry of the geometry allows us to calculate the EM fields in the quasistatic approximation. But textbook discussion do not always make completely clear that this solution to Maxwell's equations is only approximate \cite{Templin}.

The \emph{exact} solution to Maxwell's equations for a solenoid with an arbitrary time-dependent azimuthal surface current $K(t) \hat{\bm{\phi}}$ is quite complicated, but it can be solved exactly. The high symmetry of the geometry provides a nice example of how symmetries can greatly reduce the number of degrees of freedom in a problem. Surprisingly, there has been relatively little literature exploring this problem, and to the author's knowledge, no previous literature has directly solved Maxwell's equations in full generality without using the vector potential (which requires performing some very complicated surface integrals over the solenoid). This article does so.

The article proceeds as follows: In Section~\ref{Setup}, we set up the mathematical problem of the infinite ideal thin cylindrical solenoid with a time-dependent azimuthal surface current. In Section~\ref{Lit}, we review the prior literature for this problem. In Section~\ref{MaxwellEqs}, we use the symmetry of the problem to reduce Maxwell's equations (ordinary four partial differential equations for six unknown functions of four real variables) to just two coupled PDEs for two unknown functions of two real variables. The author was unable to find these simple PDEs in any of the published literature. In Section~\ref{Slow}, we consider the limiting case of slowly-varying (i.e. constant or quasistatic) currents $K(t)$. In Section~\ref{General}, we calculate the general solution for arbitrary currents. In Subsection~\ref{positive}, we discuss a class of solutions that (to the author's knowledge) has never appeared in the previous literature, because it corresponds to physically questionable (but mathematically well-defined) currents $K(t)$ that do not have a Fourier transform. In Subsection~\ref{negative}, we restrict ourselves to currents that do have a Fourier transform and work out the Green's functions for the electric and magnetic fields (\emph{not} the more common Green's function for the vector \emph{potential}) in both the frequency and the time domains, which (to the author's knowledge) have not been published in the literature before.

Most of the material in this article should be accessible to someone who has taken a graduate-level course in classical electromagnetism. Some of the material in Subsection~\ref{negative} uses some somewhat advanced concepts in harmonic analysis; in Appendix~\ref{Analytic} we give a self-contained and pedagogical explanation of the necessarily mathematical concepts. Appendix~\ref{Asymptotics} contains asymptotic forms of certain Bessel functions that are used in the main text. The main text can be read independently of the appendices if the reader is willing to skip certain proofs.

\section{Setup and assumptions} \label{Setup}

Clearly, cylindrical coordinates $(r, \phi, z)$ are the natural coordinate system to use for this cylindrical geometry.

In this article, we define an ``infinite ideal solenoid'' be be a current distribution with the following properties:
\begin{enumerate}
\item It is an infinitely long, infinitely thin cylindrical surface current with radius $R$ and surface current density $\bm{K}$.
\item The current flow is purely in the azimuthal (i.e. transverse) direction: $\bm{K} \propto \hat{\bm{\phi}}$.
\item the magnitude of the current at any given time is independent of the position on the cylinder: the magnitude $K$ is a function only of $t$ and not of $z$ or $\phi$.
\item There is no net electric (or magnetic!) charge anywhere: $\rho \equiv 0$.
\end{enumerate}

Together, these assumptions imply that $\bm{K}(\phi, z, t) \equiv K(t)\, \hat{\bm{\phi}}$. The volumetric current density $\bm{J}(r, \phi, z, t) = K(t)\,  \delta(r - R)\, \hat{\bm{\phi}}$, where $\delta$ represents the Dirac delta function.

Each of these assumptions can be translated into an idealized approximation for a physical solenoid, e.g. one made of a tightly coiled current-carrying wire. For example, the ``inifinitely thin'' part of assumption \#1 holds if the wire is very thin compared to $R$ and is coiled so tightly that we can neglect the spacing between coils. The ``infinitely long'' part holds if (a) the physical solenoid is much longer than its radius $R$, and (b) we restrict our attention to the region near the center of the cylinder, where we can neglect the fringing effects at the ends of the solenoid.\footnote{Note that once we start talking about the emitted electromagnetic fields, a new length scale will come into play: the wavelength of the electromagnetic radiation, or (closely related) the time scale over which the surface current $K(t)$ changes. The solutions that we derive below will only hold in the regime where this new length scale is far from the length scales that we are implicitly dropping in making these approximation. For example, our results below will stop being valid for wavelengths nearly as short as the diameter of the wires in the physical solenoid, or nearly as long as the length of the physical solenoid.} Assumption \#2 holds if the coils are angled nearly perpendicular to the axis of the cylinder (or if they double-wrap the cylinder with oppositely slanted layers), so that there is a negligible longitudinal current flow in the $z$-direction.

Of course, each of these simplifying assumptions can be -- and has been -- eliminated to get a more realistic model for a physical solenoid. Many, many articles consider the magnetic field for a finite-length solenoid and a steady (i.e. time-independent) current; Ref.~\onlinecite{Derby} reviews the literature and gives an exact analytic solution for thin solenoids, and Ref.~\onlinecite{Granum} compares several numerical models for both thin and thick solenoids. As Refs.~\onlinecite{Templin} and \onlinecite{Protheroe} discuss, assumptions \#2 and \#3 are very difficult to satisfy in practice: the natural way to create the current in a real solenoid is by driving one or both ends with a (potentially time-varying) voltage difference, but the signals from these driven ends can only propagate down the solenoid at the speed of light, so the instantaneous current will be different at different points along the solenoid. Ref.~\onlinecite{Protheroe} addresses this issue by considering a much more complicated model that captures the finite propagation speed of signals down the length of the solenoid and the dependence of $K$ on the $z$-coordinate. In this article, we will neglect this practical consideration for simplicity, noting that assumptions \#2 and \#3 above are perfectly consistent with special relativity, even though they are practically difficult to satisfy experimentally.

\section{Prior literature} \label{Lit}

The previous literature has extensively considered the infinite ideal solenoid with a time-varying current (and many variations). The first exact solution that the author has found was published in Ref.~\onlinecite{Abbott}, which considers a sinusoidal time-varying current $K(t) \propto \sin(\omega t)$. The authors perform some truly heroic feats of double integration over the cylinder to calculate the exact magnetic vector potential $\bm{A}$ (including all retardation effects) outside the solenoid (only), and then differentiate $\bm{A}$ to obtain exact analytic expressions for the outside electromagnetic fields in terms of Bessel functions of the first and second kind.\footnote{Ref.~\onlinecite{Abbott} also considers azimuthal current densities with time dependence $K(t)$ proportional to the Heaviside step function $K(t) \propto \Theta(t)$, $K(t) \propto t\, \Theta(t)$, and $K(t) \propto \delta(t)$, but only in the far-field limit $r \gg R$. It also considers the case of longitudinal rather than azimuthal current $\bm{K} \parallel \hat{z}$.} (Ref.~\onlinecite{Abbott} also proves a theorem that is not particularly relevant to this article, but which is so remarkable that we state it in Appendix~\ref{Vanishing} without proof.)

Ref.~\onlinecite{Templin} finds the same results by directly solving Maxwell's equations for a sinusoidal current (only). Refs.~\onlinecite{Templin, Calvert, TemplinResponse} contain an interesting back-and-forth conceptual discussion about whether the electric field \emph{outside} the solenoid is best thought of as being induced by the changing magnetic field that is located \emph{inside} or \emph{outside} of the solenoid.

We can use the Fourier transform to decompose an arbitrary current $K(t)$ into a superposition of its Fourier modes $\tilde{K}(\omega)$. Since we know the exact electromagnetic fields produced by a sinusoidal current with frequency $\omega$, we can superpose the contributions from each Fourier mode $\tilde{K}(\omega)$ to find formal expressions for the fields generated by an arbitrary current $K(t)$; these expressions for the fields are rather complicated integral transforms of $\tilde{K}(\omega)$. Ref.~\onlinecite{Bologna} expresses the fields both inside and outside the solenoid in terms of Fourier transforms.

Remarkably, the author was unable to find any published literature that takes the simplest approach of directly solving Maxwell's equations for a general current $K(t)$. (Ref.~\onlinecite{Templin} comes close, but only considers the sinusoidal case.) In this paper, we do so. Unlike in most of the literature discussed above, we do not use the vector potential $\bm{A}$; nor do we ever need to perform any geometric integrals over the solenoid surface. One of our main results simplifies the problem down to a simple pair of coupled PDEs for two functions of two variables (equations~\eqref{eqs}) that can be readily solved numerically. Another main result solves those PDEs analytically and directly expresses each electromagnetic field as a single convolution of the current $K(t)$ with an appropriate Green's function $g(r, t)$ (equations~\eqref{convs} and \eqref{grts}). Both results are exact and do not make any approximations. The intermediate steps used to derive them use various tools from harmonic analysis, but the final results do not require any Fourier transforms or even complex numbers -- only a single convolution integral for each electromagnetic field. 

\section{Maxwell's equations for the ideal solenoid} \label{MaxwellEqs}

Both inside and outside the solenoid, the electromagnetic fields satisfy the vacuum Maxwell's equations
\begin{align}
\del \cdot \bm{E} &= 0, \label{Maxwell} \\
\del \cdot \bm{B} &= 0, \nn \\
\del \cross \bm{E} &= -\frac{\partial B}{\partial t}, \nn \\
\del \cross \bm{B} &= \frac{\partial E}{\partial t}, \nn
\end{align}
where we work in Heaviside-Lorentz units where $\epsilon_0 = \mu_0 = c = 1$. At the solenoid, we have the interface boundary conditions that $\bm{E}$ is continuous,
\beq
\bm{E}_\text{in}(r = R) = \bm{E}_\text{out}(r = R), \label{Eboundary}
\eeq
and
\begin{align}
B_\text{in}^r(r = R) &= B_\text{out}^r(r = R), \label{Bboundary} \\
\bm{B}_\text{in}^\parallel(r = R) - \bm{B}_\text{out}^\parallel(r = R) &= -\bm{K} \cross \hat{\bm{r}} = K \hat{\bm{z}}, \nn
\end{align}
where $\bm{B}^\parallel$ refers to the component of $B$ that lies in the $\phi$-$z$ plane tangent to the solenoid.

The boundary conditions at spatial infinity (or more precisely, far from the solenoid as $r \to \infty$) are somewhat more complicated. If $K(t)$ vanishes identically before some starting time $t_0$, then the $\bm{E}$ and $\bm{B}$ fields vanish identically outside of the future light cone $|r - R| = t - t_0$. But if the support of $K(t)$ extends infinitely far back in the past, then we need to instead impose the condition that the electromagnetic waves propagate purely outward rather than inward in order to respect causality. There are several different ways to formalize this physical requirement mathematically, but the most common one is the \emph{Sommerfeld radiation condition} \cite{Schot}
\beq \label{Sommerfeld}
\lim_{r \to \infty} r \left( \frac{\partial}{\partial r} - i k \right) \tilde{\bm{A}}(\bm{x}, \omega) = 0,
\eeq
where $\tilde{\bm{A}}(\bm{x}, \omega)$ is the Fourier transform of the vector potential $\bm{A}(\bm{x}, t)$ with respect to time. 

A fully general electromagnetic field on flat Minkowski spacetime is represented by a two-form field $F: \mathbb{R}^4 \to \mathbb{R}^6$, and Maxwell's equations represent eight independent equations.\footnote{In $D$ spacetime dimensions, $F: \mathbb{R}^D \to \mathbb{R}^{\binom{D}{2}} = \mathbb{R}^{\frac{1}{2} D (D-1)}$, and Maxwell's equations represent  $D(D^2-3D+8)/6$ independent equations.} But in this case, we can use the symmetry of the source distribution to enormously simplify Maxwell's equations:

\begin{enumerate}
\item The time-varying current density $K(t)$ explicitly breaks time-translation and time-reversal symmetry, so the relevant symmetries are the cylindrical symmetries of the solenoid. The coordinates $\phi$ and $z$ are cyclic, so the electromagnetic fields must be independent of $\phi$ and $z$ when expressed in cylindrical coordinates, which reduces the domain from four dimensions to just two: $r$ and $t$.
\item Since the source current $\bm{K}$ is purely azimuthal and $\rho = 0$, the source four-vector $J$ is symmetric under the reflection $z \to -z$ about the $x$-$y$ plane. The $\bm{E}$ field is a true vector, so each coordinate component changes sign under an inversion of that coordinate, and so $E_z$ is odd under the reflection $z \to -z$: $E_z(z) = -E_z(-z)$ for all $z$. But using the translational invariance in $z$, $E_z(-z) = E_z(z)$, so $E_z(z) = -E_z(z)$, and so $E_z \equiv 0$ by reflectional and translational symmetry in $z$.
\item Similarly, the $\bm{B}$ field is a pseudovector: under the reflection $z \to -z$, the component $B_z$ is unchanged but the other components $\bm{B}_\parallel$ parallel to the $x$-$y$ plane change sign, and so $\bm{B}_\parallel(z) = -\bm{B}_\parallel(-z)$. But then by translational invariance in $z$, we have $\bm{B}_\parallel(-z) = \bm{B}_\parallel(z)$ and so $\bm{B}_\parallel(z) = -\bm{B}_\parallel(z)$, so $\bm{B}_\parallel(z) \equiv \bm{0}$ and $\bm{B} \parallel \hat{\bm{z}}$.
\end{enumerate}

Solely from cyclindrical and reflection symmetry, we have reduced the initial six components of the electromagnetic field down to just three ($E_r$, $E_\phi$, and $B_z$) and simplified the electromagnetic field to a much simpler map $\mathbb{R}^2 \to \mathbb{R}^3$. Since $\hat{\bm{r}}$ and $\hat{\bm{\phi}}$ are singular at the $z$-axis, the requirement that $\bm{E}$ be continuinuous there requires that $E_r$ and $E_\phi$ both approach $0$ as $r \to 0^+$. We will now drop the subscript $z$ and note that for the rest of this article, $B$ will denote the signed quantity $B_z$, not the nonnegative magnitude $|\bm{B}|$.

Using all these symmetries, Maxwell's equations~\eqref{Maxwell} reduce to

\begin{align} 
\frac{1}{r} \frac{\partial}{\partial r} (r E_r) &= 0 \label{Gauss} \\
\frac{1}{r} \frac{\partial}{\partial r} (r E_\phi) &= -\frac{\partial B}{\partial t} \nn \\
-\frac{\partial B}{\partial r} &= \frac{\partial E_\phi}{\partial t}, \nn
\end{align}
with Gauss's law for magnetism satisfied automatically.

Gauss's law for electricity~\eqref{Gauss} is easy to solve: it implies that $r E_r$ is constant both inside and outside the solenoid, and so $E_r \propto 1/r$. The requirement of continuity at the $z$-axis requires that $E_r$ vanish inside the solenoid. If there were a net surface charge on the solenoid, then the proportionality constant would be different inside and outside the solenoid -- but since we are assuming no net charge on the solenoid, the boundary condition~\eqref{Eboundary} implies that the proportionality constant is the same in both regions. So $E_r \equiv 0$ everywhere, and the $\bm{E}$ field must be purely azimuthal: $\bm{E} \parallel \hat{\bm{\phi}}$. Therefore, we have reduced the electromagnetic field to an even simpler map $\mathbb{R}^2 \to \mathbb{R}^2$. As before, we will drop the subscript on $E_\phi$ and let $E$ denote the signed quantity $E_\phi$, not the magnitude $|\bm{E}|$.

We are finally left with the quite simple coupled partial differential equations
\begin{subequations} \label{eqs}
\begin{empheq}[box=\fbox]{align}
\frac{1}{r} \frac{\partial}{\partial r} (r E) = \frac{E}{r} + \frac{\partial E}{\partial r} &= -\frac{\partial B}{\partial t} \label{Faraday} \\
-\frac{\partial B}{\partial r} &= \frac{\partial E}{\partial t}. \label{Ampere}
\end{empheq}
\end{subequations}
Equation~\eqref{Faraday} is Faraday's law, and equation~\eqref{Ampere} is Amp\`ere's law. The real-valued fields $E(r, t)$ and $B(r, t)$ are defined on the half-plane $(r \geq 0, t \in \mathbb{R})$. The interface boundary conditions~\eqref{Eboundary} and~\eqref{Bboundary} and the requirement of continuity reduce to the requirements that
\begin{enumerate}
\item $E$ be continuous on the whole half-plane $r \geq 0$,
\item $E \to 0$ as $r \to 0^+$,
\item $B$ have a jump discontinuity
\[
B(r \to R^-, t) - B(r \to R^+, t) = K(t)
\]
along the coordinate line $r = R$, and
\item $E$ and $B$ respect the boundary conditions discussed below equations~\eqref{Bboundary} as $r \to \infty$.
\end{enumerate}

The equations~\eqref{eqs} are the first main result of this article. To the author's knowledge, they have not been explicitly published in any previous textbooks or scientific journals (although they can be derived from the results in Ref.~\onlinecite{Templin}).  This system of equations, which is exactly equivalent to Maxwell's equations for the ideal solenoid, is probably the best formulation of the problem to use for purely numerical methods. As a pair of coupled linear homogeneous first-order PDEs in only two variables, they are relatively straightforward to solve numerically. Note that nothing up to this point has required any complications such as gauge potential fields, complex numbers, or even any integrals.

\section{Slowly-varying currents} \label{Slow}

Let us check the familiar textbook situations in which the currents are static or quasistatic.

\subsection{Constant current} \label{constant}

Suppose that $K(t) \equiv K$ is constant. In this case, we can solve equations~\eqref{eqs} by inspection. Since the boundary conditions are time-independent, it is natural to guess that $E$ and $B$ are constant inside and outside the solenoid, so that Amp\`ere's law~\eqref{Ampere} is trivially satisfied. Faraday's law~\eqref{Faraday} implies that $E \propto 1/r$, and then boundary condition \#2 requires that $E \equiv 0$. Boundary condition \#4 requires that $B_\text{out} \equiv 0$ outside the solenoid, and then boundary condition \#3 requires that $B_\text{in} \equiv K$ inside. This is the familiar exact magnetostatic solution for an ideal solenoid with constant current.

\subsection{Quasistatic current}

When first introducing Faraday's law, EM courses often discuss the example of a quasistatic solenoid -- although they do not always make fully clear that this is an approximate rather than an exact solution to Maxwell's equations \cite{Templin}. Textbooks often explain that this quasistatic solution holds when the current is ``slowly varying'', although they do not always explain exactly what this means (i.e., slow compared to what?).

In this approximation, the magnetic field is assumed to be spatially uniform inside the solenoid with a magnitude that instantaneously matches the solenoid current ($B_\text{in}(r,t) \equiv K(t)$) and to vanish outside ($B_\text{out} \equiv 0$). The electric field is assumed to be $E_\text{in} = -\frac{1}{2} \frac{dK}{dt} r$ inside the solenoid and $E_\text{out} = -\frac{1}{2} \frac{dK}{dt} \frac{R^2}{r}$ outside, in accordance with the integral form of Faraday's law. Faraday's law~\eqref{Faraday} is exactly satisfied, as are the interface boundary conditions \#1-3,\footnote{There are some conceptual difficulties around the meaning of the Sommerfeld radiation condition in this case, since Maxwell's equations are not satisfied.} but Amp\`ere's law~\eqref{Ampere} is not: the LHS $-\frac{\partial B}{\partial r}$ vanishes, but the RHS
\[
\frac{\partial E}{\partial t} = \frac{1}{2} \frac{d^2K}{dt^2} r \neq 0\]
inside the solenoid and
\[
\frac{\partial E}{\partial t} = \frac{1}{2} \frac{d^2K}{dt^2} \frac{R^2}{r} \neq 0\]
outside. In this case the quasistatic approximation is ``non-perturbative'', since we are not dropping a subleading correction term but instead the \emph{only} nonzero term. Noting that $\frac{d^2K}{dt^2}$ is the key parameter that set the scale in Amp\`ere's law~\eqref{Ampere}, we must resort to dimensional analysis and conclude that the quasistatic approximation holds if the dimensionless ratio
\[
\frac{R^2}{K(t)} \frac{d^2K}{dt^2} \ll 1
\]
at all times.

As noted in Refs.~\onlinecite{Abbott, Calvert}, if $K(t)$ depends linearly on time, then $\frac{d^2K}{dt^2} \equiv 0$, the quasistatic approximation becomes exact, and the magnetic field vanishes identically outside of the solenoid. It may seem somewhat counterintuitive that in this case, Maxwell's equations are \emph{exactly} solved by using the Biot-Savart law with the time-varying current source evaluated at the \emph{present} time rather than the retarded time, but nevertheless it is true \cite{GriffithsHeald}. The situation is somewhat reminiscent of the fact that the electric field created by a relativistic electric point charge moving with constant velocity points exactly at the charge's \emph{present} position, rather than at its retarded position as we might expect. In any case, a surface current that depents linearly on $t$ for all time is somewhat artificial, since it becomes unboundedly large in the far past and far future.

\section{The general solution} \label{General}

The partial differential equations~\eqref{eqs} are in the most useful form for numerical solutions, because coupled first-order differential equations are usually more convenient to tackle numerically than uncoupled second-order differential equations. Nevertheless, we can decouple the equations~\eqref{eqs} by taking the partial derivative of \eqref{Faraday} with respect to $r$, taking the partial derivative of \eqref{Ampere} with respect to $t$, and equating them to get the decoupled PDE for $E$
\beq \label{E}
\frac{\partial}{\partial r} \left( \frac{1}{r} \frac{\partial}{\partial r} (r E) \right) = \frac{\partial^2 E}{\partial t^2}.
\eeq
Similarly, we can take the partial derivative of \eqref{Faraday} with respect to $t$, apply the differential operator
\[
\frac{1}{r} \frac{\partial}{\partial r} (r \underline{\phantom{x}}) = \frac{1}{r} + \frac{\partial}{\partial r}
\]
to \eqref{Ampere}, and equate them to get the decoupled PDE for $B$
\beq \label{B}
\frac{1}{r} \frac{\partial}{\partial r} \left( r \frac{\partial B}{\partial r} \right) = \frac{\partial^2 B}{\partial t^2}.
\eeq
Equations~\eqref{E} and \eqref{B} are somewhat similar -- but not identical -- to the usual 1D wave equation.

Before we dive into the messy details of solving these PDEs, we will outline our plan of attack. The typical textbook approach for dealing with an uncoupled linear PDE is to first make a separation-of-variables ansatz that breaks the PDE into an equality between decoupled ODEs of different dependent variables, which must equal a common unknown constant. The boundary conditions typically restrict the values that the constant can take on, and the full solution is a linear combination over the allowed constants, with coefficients determined by the boundary conditions.

Plugging in the ansatzes $E(r,t) = \mathcal{R}_E(r) T_E(t),\ B(r,t) = \mathcal{R}_B(r) T_B(t)$ and separating variables gives the eigenvalue equations
\begin{align}
\frac{1}{\mathcal{R}_E} \frac{d}{dr} \left( \frac{1}{r} \frac{d}{dr} (r \mathcal{R}_E) \right) &= \frac{1}{T_E} \frac{d^2 T_E}{dt^2} = \lambda, \label{separated} \\
\frac{1}{r \mathcal{R}_B} \frac{d}{dr} \left( r \frac{d \mathcal{R}_B}{dr} \right) &= \frac{1}{T_B} \frac{d^2 T_B}{dt^2} = \lambda. \nn
\end{align}
$\mathcal{R}_E(r)$ must be continuous everywhere by boundary condition \#1, but it can have a kink (i.e. fail to be differentiable) at $r = R$.\footnote{The letter $\lambda$ denotes an eigenvalue, \emph{not} a wavelength. The eigenfunctions $\mathcal{R}$ and $T$ implicitly depend on the eigenvalue $\lambda$, but we will omit this dependence in order to simplify the notation. A priori, the eigenvalue $\lambda$ could take on different values $\lambda_E$ and $\lambda_B$ across the two equations~\eqref{separated}. But as we will see below, the coupled equations \eqref{eqs} require that $\lambda_{E,\text{in}} = \lambda_{B,\text{in}}$ and $\lambda_{E,\text{out}} = \lambda_{B,\text{out}}$ in order for the time dependence to have the same frequency. And the requirement that $E$ be continuous at the solenoid $r = R$ for all $t$ requires that $\lambda_{E,\text{in}} = \lambda_{E,\text{out}}$, so $\lambda_E = \lambda_B \equiv \lambda$ everywhere.}

In general, the eigenvalues $\lambda$ can range over a continuous set of real values. Therefore, we will need to solve equations~\eqref{separated} for each possible eigenvalue $\lambda$ -- which is already quite nontrivial -- and then superpose them not through a discrete linear combination, but through an integral over the allowed eigenvalues $\lambda$.

\subsection{Case 1: $\lambda = 0$.}
This case is the simplest. It turns out to yield exactly the quasistatic solution
\begin{align}
&K(t) = K_0 + \dot{K} t, \label{quasistatic} \\
&B_\text{in}(r, t) \equiv K(t), &&B_\text{out}(r, t) \equiv 0, \nn \\
&E_\text{in}(r, t) = -\frac{1}{2} \dot{K} r, &&E_\text{out}(r, t) = -\frac{1}{2} \dot{K} \frac{R^2}{r} \nn
\end{align}
discussed in the previous section, where $\dot{K}$ is a \emph{constant}. This is the contribution to the electromagnetic fields from any \emph{linear} component of $K(t)$ -- although as we mentioned above, it is not clear how physically realistic is a surface current that gets unboundedly large in the far past and far future.

\subsection{Case 2: $\lambda > 0$.} \label{positive}
This case is even more physically questionable. It leads to mathematically valid solutions to Maxwell's equations (whose radial functions $\mathcal{R}_E(r)$ and $\mathcal{R}_B(r)$ are modified Bessel functions) that technically satisfy all four boundary conditions listed above. However, these solutions correspond to surface currents $K(t)$ that grow exponentially with time in the far past and/or far future. As such, the function $K(t)$ is not a tempered distribution and does not have a well-defined Fourier transform -- not even in the distributional sense -- and so much of our analytical machinery breaks down. The author is not aware of any previous literature that mentions the existence of these formal solutions to Maxwell's equations for the ideal solenoid, since most previous literature has implicitly assumed that $K(t)$ is a tempered distribution with a well-defined Fourier transform. We will not work out these physically unrealistic solutions to Maxwell's equations in detail, but it is interesting to note their existence.

Moreover, it is interesting to consider the example of a surface current $K(t)$ that starts from $0$ in the infinite past and then grows exponentially (with some growth constant $\tau$) up until some time $t_\text{change}$, when it then changes to a different time dependence that stays bounded over future times. This current profile is reasonably physically realistic. Of course, we could not experimentally arrange for it told hold exactly out to the infinite past -- but it is negligibly small in the far past, so we could closely approximate it by turning on a very small current and then letting the current grow exponentially for several $e$-folding times $\tau$ before we taper off its growth.

In this example, the spacetime region $|r - R| > t - t_\text{change}$ lies outside the future light cone of the ``new'' current. Therefore, the past light cones of every point in this spacetime region only intersect with the exponentially growing current. Therefore, this region has no way to ``know'' that the current will eventually change from exponential growth, and so the electromagnetic fields in this region are the same as those produced by a hypothetical (and physically unrealistic) current that grows exponentially for \emph{all} time. In this spacetime region (only), the radial functions $\mathcal{R}_E$ and $\mathcal{R}_B$ are given \emph{exactly} by \emph{modified} Bessel functions that take the dimensionless argument $r/\tau$.

\subsection{Case 3: $\lambda < 0$.} \label{negative}
This is the most physically relevant case. It will be convenient to change variables to $\lambda = -\omega^2$, where $\omega \in \mathbb{R} \setminus \{0\}$. We see that $T(t)$ is sinusoidal with frequency $|\omega|$. We can introduce the nondimensionalized variable\footnote{Since we are working in vacuum and in units where $c=1$, the wave number $k$ equals the frequency $\omega$. Note that this substitution fails if $\omega = 0$, so taking the limit $\omega \to 0$ of the expressions derived in this subsection does \emph{not} yield the quasistatic solution~\eqref{quasistatic}, but only the static solution discussed in subsection~\ref{constant}.} $x := \omega r$ and rearrange the radial equations to
\begin{align*}
\frac{d^2 \mathcal{R}_E}{dx^2} + \frac{1}{x} \frac{d \mathcal{R}_E}{dx} + \left( 1 - \frac{1}{x^2} \right) \mathcal{R}_E &= 0, \label{RE} \\
\frac{d^2 \mathcal{R}_B}{dx^2} + \frac{1}{x} \frac{d \mathcal{R}_B}{dx} + \mathcal{R}_B &= 0
\end{align*}
which are Bessel's equations of order 1 and 0 respectively.

There is no single Bessel function that satisfies both the boundary conditions \#2 and \#4, so $\mathcal{R}_E(x)$ must have a kink (or higher-derivative discontinuity) at $x = \omega R$, and $\mathcal{R}_B(x)$ must have a jump discontinuity there by boundary condition \#3. Boundary condition \#2 requires that $\lim \limits_{x \to 0^+} \mathcal{R}_E(x) = 0$, and continuity requires that $\lim \limits_{x \to 0^+} \mathcal{R}_B(x)$ must converge to a finite value. These requirement are satisfied by the Bessel functions of the first kind $\mathcal{R}_E(x) \propto J_1(x)$ and $\mathcal{R}_B(x) \propto J_0(x)$ inside the solenoid.

Boundary condition \#4 requires that the waves outside the solenoid be purely outgoing. If $\omega > 0$, then this means that they must have the form
\begin{align*}
E_\text{out} &\propto \Re \left[ e^{i(-\omega t + \delta_E)} H^{(1)}_1(\omega r) \right], \\
B_\text{out} &\propto \Re \left[ e^{i (-\omega t + \delta_B)} H^{(1)}_0(\omega r) \right],
\end{align*}
where $H^{(1)}_\alpha(x)$ are the Hankel functions of the first kind and the implicit proportionality constants can depend on $\omega$ (but not on $r$ or $t$).\footnote{These are the standard Bessel functions, not the modified Bessel functions mentioned above for the unphysical $\lambda > 0$ case.} If $\omega < 0$, then the outward-propagating waves instead correspond to the Hankel functions of the second kind $H^{(2)}_\alpha(x)$.

We can now play the usual game of assigning coefficients to each of the four functions $E_\text{in}$, $E_\text{out}$, $B_\text{in}$, and $B_\text{out}$ and then using equations~\eqref{eqs} and the boundary conditions to fix their values. The process is straightforward but requires some algebra and some obscure Bessel function identities.

Ref.~\onlinecite{Templin} worked it out for the special case of a sinusoidal surface current $K_\omega(t)$ oscillating at a single positive frequency $\omega$. If we characterize the sinusoidal current $K_\omega(t)$ by a \emph{complex} constant $\tilde{K}_\omega$ via\footnote{Ref.~\onlinecite{Templin} only considered the case of positive $K$ and $K(t) \propto \cos(\omega t)$, but it easy to generalize the author's results to the case where $K$ is complex and $K(t)$ has a phase offset.}
\beq \label{Komega}
K_\omega(t) = \Re \big[ \tilde{K}_\omega\, e^{-i \omega t} \big] = |\tilde{K}_\omega| \cos \big( \omega t - \arg{\tilde{K}_\omega} \big),
\eeq
then the solution works out to be
\begin{widetext}
\begin{align}
E_\text{in}^\omega(r, t) &= -\frac{1}{2} \pi \omega R\, J_1(\omega r) \Re \left[ \phantom{i} \tilde{K}_\omega H^{(1)}_1(\omega R)\, e^{-i \omega t} \right], &r &\in [0, R]  \label{omegasols} \\
E_\text{out}^\omega(r, t) &= -\frac{1}{2} \pi \omega R J_1(\omega R) \Re \left[ \phantom{i} \tilde{K}_\omega H^{(1)}_1(\omega r\, )\, e^{-i \omega t} \right], &r &\geq R \nn \\
B_\text{in}^\omega(r, t) &= \phantom{-}\frac{1}{2} \pi \omega R J_0(\omega r\, ) \Re \left[ i \tilde{K}_\omega H^{(1)}_1(\omega R)\, e^{-i \omega t} \right], &r &\in [0, R] \nn \\
B_\text{out}^\omega(r, t) &= \phantom{-}\frac{1}{2} \pi \omega R J_1(\omega R) \Re \left[ i \tilde{K}_\omega H^{(1)}_0(\omega r\, )\, e^{-i \omega t} \right], &r &\geq R. \nn
\end{align}
\end{widetext}
For negative $\omega$, all the Hankel functions of the first kind $H^{(1)}_\alpha$ in eqs.~\eqref{omegasols} are replaced by Hankel functions of the second kind $H^{(2)}_\alpha$.

For a general (tempered-distribution) surface current $K(t)$, we can superpose the fields~\eqref{omegasols} sourced by a single Fourier mode~\eqref{Komega} over all the Fourier modes with amplitude (and phase) specified by $\tilde{K}(\omega)$ at frequency $\omega$. But in this case, doing so turns out to be somewhat subtle; Ref.~\onlinecite{Bologna} is the only source the author could find that actually discusses how to do so.\footnote{Griffiths' textbook \cite{Griffiths} only says ``Any wave can be written as a linear combination of sinusoidal waves, and therefore you know how sinusoidal waves behave, you know in principle how \emph{any} wave behaves.'' Similarly, Jackson's textbook \cite{Jackson} only says ``Assum[e] solutions with harmonic time dependence $e^{-i \omega t}$, from which we can build an arbitrary solution by Fourier superposition.''}

The difficulty is that, as mentioned above, for positive $\omega$ the outward-propagating waves correspond to the Hankel functions of the first kind $H^{(1)}_\alpha(x)$, but for negative $\omega$ they correspond to the Hankel functions of the second kind $H^{(2)}_\alpha(x)$. There are two possible ways to deal with this difficulty.

The first approach is to decompose $K(t)$ via the usual Fourier transform
\[
K(t) = \frac{1}{2\pi} \int \limits_{-\infty}^\infty \tilde{K}(\omega) e^{-i \omega t}\, d\omega,
\]
where $\omega$ ranges over all of $\mathbb{R}$. Under this approach, we can find the general fields $E_\text{in}$, $E_\text{out}$, $B_\text{in}$, and $B_\text{out}$ by promoting the constant $\tilde{K}_\omega$ in \eqref{omegasols} to the function $\tilde{K}(\omega)$ and then integrating over all $\omega \in \mathbb{R}$. This is the approach implicitly taken by Ref.~\onlinecite{Bologna}. The advantage of this first option is that we can use all the familiar machinery of Fourier analysis; the disadvantage is that the integrand is complicated and is discontinuous at $\omega = 0$, because we need to define it piecewise using expressions~\eqref{omegasols} for $\omega > 0$ and the corresponding expressions with $H^{(2)}_\alpha$ for $\omega < 0$.

The second approach is to instead decompose $K(t)$ using the modified inverse Fourier transform
\beq \label{FourierK}
K(t) = \Re \left[ \frac{1}{\pi} \int \limits_0^\infty \tilde{K}(\omega) e^{-i \omega t}\, d\omega \right],
\eeq
where now we only integrate over positive $\omega$. It turns out that since $K(t)$ is real, this modified Fourier decomposition holds if we use the Fourier coefficients $\tilde{K}(\omega)$ given by the usual forward Fourier transform
\[
\tilde{K}(\omega) = \int \limits_{-\infty}^\infty K(t)\, e^{i \omega t}\, dt.
\]
Loosely speaking, we need to use a prefactor $\frac{1}{\pi}$ in eq.~\eqref{FourierK} instead of the usual $\frac{1}{2\pi}$ because there is an implicit factor of $2$ that corrects for the fact that we are only integrating over positive frequencies. See Appendix~\ref{Analytic} for a more rigorous proof of eq.~\eqref{FourierK} and a conceptual explanation of why it works. The advantage of this second option is that the integrands are simpler and continuous; the disadvantage is that we need to use slightly more complicated Fourier machinery to generalize results like the convolution theorem.

Since it is already much more convenient to express the integrands \eqref{omegasols} as the real parts of complex functions, we will choose the second option. We can find the general fields $E_\text{in}$, $E_\text{out}$, $B_\text{in}$, and $B_\text{out}$ by promoting the constant $\tilde{K}_\omega$ in \eqref{omegasols} to the function $\tilde{K}(\omega)$ and then integrating over only positive $\omega$:
\begin{align} \label{IFTsols}
E_\text{in}(r,t) &= \frac{1}{\pi} \int_0^\infty E_\text{in}(r,t;\omega)\, d\omega \\
E_\text{out}(r,t) &= \frac{1}{\pi} \int_0^\infty E_\text{out}(r,t;\omega)\, d\omega \nn \\
B_\text{in}(r,t) &= \frac{1}{\pi} \int_0^\infty B_\text{in}(r,t;\omega)\, d\omega \nn \\
B_\text{out}(r,t) &= \frac{1}{\pi} \int_0^\infty B_\text{in}(r,t;\omega)\, d\omega. \nn
\end{align}

If we plug the expressions~\eqref{omegasols} into equations~\eqref{IFTsols}, then we see that
\beq \label{Einrt}
E_\text{in}(r,t) = \Re \left[ \frac{1}{\pi} \int_0^\infty \tilde{g}_{E_\text{in}}(r, \omega)\, \tilde{K}(\omega)\, e^{-i \omega t}\, d\omega \right]
\eeq
and similarly for the other three fields, where
\begin{align}
\tilde{g}_{E_\text{in}}(r, \omega) &:= \frac{-1}{2} \pi\, \omega R\, J_1(\omega r\, )\, H^{(1)}_1(\omega R), \label{gtilde} \\
\tilde{g}_{E_\text{out}}(r, \omega) &:= \frac{-1}{2} \pi\, \omega R\, J_1(\omega R)\, H^{(1)}_1(\omega r), \nn \\
\tilde{g}_{B_\text{in}}(r, \omega) &:= \frac{+i}{2} \hspace{1pt} \pi\, \omega R\, J_0(\omega r\, )\, H^{(1)}_1(\omega R), \nn \\
\tilde{g}_{B_\text{out}}(r, \omega) &:= \frac{+i}{2} \hspace{1pt} \pi\, \omega R\, J_1(\omega R)\, H^{(1)}_0(\omega r). \nn
\end{align}

Ref.~\onlinecite{Bologna} derives a modified version of equations~\eqref{Einrt} and \eqref{gtilde}, but does not consider the time-domain duals of the frequency-domain Green's functions~\eqref{gtilde}. We will derive the time-domain Green's functions in the rest of this section, which will allow us to directly solve the problem in the time domain without every needing to go into the frequency domain (or needing to use any complex numbers).

Equation~\eqref{Einrt} is similar to (the real part of) an inverse Fourier transform, but not identical because we only integrate over positive frequencies. If the integral in eq.~\eqref{Einrt} extended over all $\omega \in \mathbb{R}$ (and the prefactor were $\frac{1}{2\pi}$), then we would be able to apply the convolution theorem, and we would find that $E_\text{in}(r, t) = \Re[g_{E_\text{in}}(r, t) \ast K(t)] = \Re[g_{E_\text{in}}(t)] \ast K(t)$ because $K(t)$ is real, where $g_{E_\text{in}}(r, t)$ would be the inverse Fourier transform of $\tilde{g}_{E_\text{in}}(r, \omega)$. So $\Re[g_{E_\text{in}}(r, t)]$ would be the time-domain Green's function for the electric field inside the solenoid. 

The fact that the integral in equation~\eqref{Einrt} only runs over positive frequencies means that we cannot apply the standard convolution theorem. Nevertheless, as we prove in Corollary~\ref{1sidedconv} of Appendix~\ref{Analytic}, the argument goes through with only minimal changes, and
\begin{empheq}[box=\fbox]{align}
E_\text{in}(r, t) &= g_{E_\text{in}}(r, t) \ast K(t) \label{convs} \\
E_\text{out}(r, t) &= g_{E_\text{out}}(r, t) \ast K(t) \nn \\
B_\text{in}(r, t) &= g_{B_\text{in}}(r, t) \ast K(t) \nn \\
B_\text{out}(r, t) &= g_{B_\text{out}}(r, t) \ast K(t), \nn
\end{empheq}
where all convolutions are taken with respect to time. Equations~\eqref{convs} are the second main result of this article.  Equation~\eqref{reonesidedconv} gives that since $K(t)$ is real, each of the four Green's functions has the form
\beq \label{modifiedIFT}
g(r,t) = \mathrm{Re} \left[ \frac{1}{\pi} \int_0^\infty \tilde{g}(r, \omega)\, e^{-i \omega t}\, d\omega \right]
\eeq
where $\tilde{g}(r, \omega)$ denotes the corresponding frequency-domain expression~\eqref{gtilde}. Explicitly, we get the time-domain Green's functions\footnote{Note that $\tilde{g}_{E_\text{in}}(r = 0, \omega) \equiv g_{E_\text{in}}(r = 0, t) \equiv 0$, so $E_\text{in}(r = 0, t) \equiv 0$ in accordance with boundary condition \#2.}
\begin{empheq}[box=\fbox]{align}
g_{E_\text{in}}(r, t) &= \mathrm{Re} \left[ \frac{-1}{2} \int_0^\infty \omega R\, J_1(\omega r\, )\, H^{(1)}_1(\omega R)\, e^{-i \omega t}\, d\omega \right] \label{grts} \\
g_{E_\text{out}}(r, t) &= \mathrm{Re} \left[ \frac{-1}{2} \int_0^\infty \omega R\, J_1(\omega R)\, H^{(1)}_1(\omega r)\, e^{-i \omega t}\, d\omega \right] \nn \\
g_{B_\text{in}}(r, t) &= \mathrm{Re} \left[ \frac{+i}{2} \int_0^\infty \omega R\, J_0(\omega r\, )\, H^{(1)}_1(\omega R)\, e^{-i \omega t}\, d\omega \right] \nn \\
g_{B_\text{out}}(r, t) &= \mathrm{Re} \left[ \frac{+i}{2} \int_0^\infty \omega R\, J_1(\omega R)\, H^{(1)}_0(\omega r)\, e^{-i \omega t}\, d\omega \right]. \nn
\end{empheq}

Note that in the special case where $K(t) \equiv K$ is constant, equations~\eqref{convs} and \eqref{modifiedIFT} give that
\begin{align*}
E(r, t) &\equiv K \int_{-\infty}^\infty g(r, \tau)\, d\tau \\
&= K\, \mathrm{Re} \left[ \frac{1}{\pi} \int_{-\infty}^\infty \int_0^\infty \tilde{g}(r, \omega)\, e^{-i \omega \tau}\, d\omega\, d\tau \right] \\
&= K\, \mathrm{Re} \left[ \frac{1}{\pi} \int_0^\infty \tilde{g}(r, \omega) \left( \int_{-\infty}^\infty e^{-i \omega \tau}\, d\tau \right) d\omega \right] \\
&= K\, \mathrm{Re} \left[ 2 \int_0^\infty \tilde{g}(r, \omega)\, \delta(\omega)\, d\omega \right] \\
&= K\, \mathrm{Re} \left[ \tilde{g}(r, \omega = 0) \right],
\end{align*}
where we interpret $\tilde{g}(r, \omega = 0)$ to mean $\lim \limits_{\omega \to 0^+} \tilde{g}(r, \omega)$ since the integral only runs over nonnegative values of $\omega$.
We prove in equations~\eqref{lowfreqgs} of Appendix~\ref{Asymptotics} that
\begin{align*}
\lim_{\omega \to 0^+} \tilde{g}_{E_\text{in}}(r, \omega) = \lim_{\omega \to 0^+} \tilde{g}_{E_\text{out}}(r, \omega) = \lim_{\omega \to 0^+} \tilde{g}_{B_\text{out}}(r, \omega) &= 0, \\
\lim_{\omega \to 0^+} \tilde{g}_{B_\text{in}}(r, \omega) &= 1.
\end{align*}
We therefore get the familiar magnetostatic solution $E(r, t) \equiv 0$, $B(r, t) = K\, \theta(R - r)$.

Since all four frequency-domain Green's functions~\eqref{gtilde} remain bounded in the low-frequency limit $\omega \to 0^+$, all four integrals in equations~\eqref{grts} converge at the lower ends of the intervals of integration. But as we prove in Appendix~\ref{Asymptotics}, in the high-frequency limit all four frequency-domain Green's functions~\eqref{gtilde} asymptotically approach
\beq \label{ghighfreq}
\tilde{g}(r, \omega) \sim \tilde{g}_\text{HF}(r, \omega) := \frac{1}{2} \sqrt{\frac{R}{r}} \left[ \pm e^{i\, |r-R|\, \omega} + e^{i \left( (R+r) \omega - \frac{1}{2} \pi \right)} \right]
\eeq
as $\omega \to +\infty$, where we choose the $+$ branch of the $\pm$ symbol for $\tilde{g}_{B_\text{in}}$ and the $-$ branch for $\tilde{g}_{E_\text{in}}$, $\tilde{g}_{E_\text{out}}$, and $\tilde{g}_{B_\text{out}}$.\footnote{Just as the magnetic field $B_\text{in}$ behaves differently from the other three fields in the static case $\omega = 0$ -- namely, by being nonzero -- it also behaves differently in the high-frequency limit $\omega r \gg 1$. Along the $r = 0$ axis, $\tilde{g}_{E_\text{in}}(r = 0, \omega) \equiv 0$ and $\tilde{g}_{B_\text{in}}(r = 0, \omega) \sim \sqrt{\frac{\pi}{2} \omega R}\, e^{i \left( \omega R - \frac{1}{4} \pi \right)}$ as $\omega \to +\infty$. We neglect this measure-zero case for the rest of this article; we can handle it using a similar approach as in the $r > 0$ case, but doing so complicates the results by adding more special cases that need to be handled separately.} Since the frequency-domain Green's functions do not go to zero in the high-frequency limit $\omega \to +\infty$, they are not square-integrable, and so the integrals~\eqref{grts} do not converge to proper functions $g(r, t)$ but to distributions that contain Dirac delta functions.

To separate out the ``well-behaved'' contributions to the Green's functions~\eqref{grts}, we decompose each of the four frequency-domain Green's functions~\eqref{gtilde} into the sum of the asymptotic exponential $\tilde{g}_\text{HF}(r, \omega)$ given in equation~\eqref{ghighfreq} and a remainder term 
\beq \label{gtildeSI}
\tilde{g}_\text{SI}(r, \omega) := \tilde{g}(r, \omega) - \tilde{g}_\text{HF}(r, \omega)
\eeq
that will turn out to be square-integrable:
\beq \label{freqdecomp}
\tilde{g}(r, \omega) = \tilde{g}_\text{SI}(r, \omega) + \tilde{g}_\text{HF}(r, \omega).
\eeq
It turns out that the $\tilde{g}_{B,\text{HF}}(r, \omega)$ terms~\eqref{ghighfreq} for the $B$ field contain the entire jump discontinuity of $\tilde{g}_B(r, \omega)$ at $r = R$, so the remaining square-integrable term $\tilde{g}_{B,\text{SI}}(r, \omega)$ is continuous at $r = R$. We can distribute the modified inverse Fourier transform~\eqref{modifiedIFT} across the two terms of the decomposition~\eqref{freqdecomp} and define
\begin{align}
g_\text{SI}(r, t) &:= \mathrm{Re} \left[ \frac{1}{\pi} \int_0^\infty \tilde{g}_\text{SI}(r, \omega)\, e^{-i \omega t}\, d\omega \right], \label{gSI} \\
g_\text{HF}(r, t) &:= \mathrm{Re} \left[ \frac{1}{\pi} \int_0^\infty \tilde{g}_\text{HF}(r, \omega)\, e^{-i \omega t}\, d\omega \right], \label{gHF}
\end{align}
with $g(r, t) = g_\text{SI}(r, t) + g_\text{HF}(r, t)$.\footnote{The subscript ``HF'' (for ``high-frequency'') in $g_\text{HF}(r,t)$ should not be interpreted too literally. It is simply a mnemonic for the inverse Fourier transform of the particular choice of high-frequency limiting expression $\tilde{g}_\text{HF}(r, \omega)$ of $\tilde{g}(r, \omega)$ that is given in equation~\eqref{ghighfreq} -- but $g_\text{HF}(r, t)$ includes the low-frequency modes in $\tilde{g}_\text{HF}(r, \omega)$. The particular choice of function~\eqref{ghighfreq} is not the unique asymptotic function for $\tilde{g}(r, \omega)$, but it is the most convenient choice for our purposes. Therefore, the particular choice of decomposition~\eqref{freqdecomp} is not a unique decomposition of $\tilde{g}(r, \omega)$ into the sum of a ``square-integrable part'' and a ``high-frequency part''.}

We can calculate the second modified inverse Fourier transform~\eqref{gHF} to express $g_\text{HF}(r, t)$ in closed form. We prove in equation~\eqref{OnesidedExpIFTapp} of Appendix~\ref{Analytic} that if $A \in \mathbb{C}$ and $t_0 \in \mathbb{R}$, then 
\beq \label{OnesidedExpIFTbody}
\Re \left[ \frac{1}{\pi} \int_0^\infty A\, e^{i \omega t_0}\, e^{-i \omega t} \, d\omega \right] = \Re[A]\, \delta(t - t_0) + \frac{\Im[A]}{\pi (t - t_0)}.
\eeq
Both terms of $\tilde{g}_\text{HF}(r, \omega)$ in equation~\eqref{ghighfreq} take the form $A\, e^{i \omega t_0}$ with $A = \pm \frac{1}{2} \sqrt{\frac{R}{r}},\ t_0 = |r-R|$ and $A = -\frac{1}{2} \sqrt{\frac{R}{r}}\, i, t_0 = R + r$, respectively. Therefore, equations~\eqref{ghighfreq}, \eqref{gHF}, and \eqref{OnesidedExpIFTbody} give that
\beq \label{gHFrt}
g_\text{HF}(r, t) = \frac{1}{2} \sqrt{\frac{R}{r}} \left[ \pm \delta(t - |r - R|) - \frac{1}{\pi (t - (R + r))} \right].
\eeq
It turns out that $g_E(r, t)$ is continuous across the solenoid $r = R$, and $g_B(r, t)$ is continuous except for the sign of the delta-function term, which switches across the solenoid. 

Returning to equations~\eqref{convs}, we see that the first term in equation~\eqref{gHFrt} (the delta function) corresponds to the contribution to the electromagnetic fields from the wavefront that is just arriving at the point $(r, t)$ after having been emitted at the solenoid at time $t - |r - R|$ and traveled radially inward or outward at the speed of light. This delta-function term represents the very first influence that a given point ``feels'' as the initial wavefront arrives. The second term is inversely proportional to $t-(R + r)$, which is the retarded time at which the \emph{back} end of the solenoid (antipodal to the observation point) emitted the EM waves that reach the observation point at time $t$. 

Causality requires that the time-domain Green's functions $g(r, t)$ all vanish outside the future light cone $t = |r - R|$.\footnote{For this geometry, the radius $r$ represents a cylinder rather than a point, so the future light cone $t = |r - R|$ is a three-dimensional null hypersurface. It starts out aligned with the solenoid (as a cylinder of radius $R$) and then splits into two concentric cylinders. One cylinder shrinks in radius at the speed of light before vanishing at time $t = R$ when it degenerates at the solenoid axis, and the other increases in radius without bound at the speed of light. The ``outside'' of this light cone corresponds to the region $|R - r| > t$ that is spacelike separated from the entire solenoid. For $t < R$, this region consists of two disconnected components, one bounded in radius; for $t > R$, it consists of a single connected component unbounded in radius.} Therefore, outside the light cone we must have that $g_\text{SI}(r, t) \equiv -g_\text{HF}(r, t)$ for all four Green's functions. From equations~\eqref{gSI} and \eqref{gHFrt}, we see that causality implies the purely mathematical fact that
\beq \label{causality}
\mathrm{Re} \left[ \int_0^\infty \tilde{g}_\text{SI}(r, \omega)\, e^{-i \omega t}\, d\omega \right] = \frac{\sqrt{R/r}}{2 (t - (R + r))}
\eeq
if $0 < t < |R - r|$ for all four functions~\eqref{gtilde}. This mathematical result is easy to check numerically, but the author is not aware of any way to prove it directly without using this causality argument.

The modified inverse Fourier transform~\eqref{gSI} for $g_\text{SI}(r, t)$ cannot (to the author's knowledge) be computed in closed form for any of the four electromagnetic fields described in equations~\eqref{grts}. But since the frequency-domain functions $\tilde{g}_\text{SI}(r, \omega)$ defined in equation~\eqref{gtildeSI} are all square-integrable, we can calculate the Fourier transforms numerically (although the fact that the integrands are oscillating means that we need to be careful with numerical stability).

In Figure~\ref{plots}, we plot the time-domain Green's functions $g(r, t) = g_\text{SI}(r,t) + g_\text{HF}(r, t)$ for the electric and magnetic fields inside and outside the solenoid. As expected, the Green's functions vanish identically outside of the solenoid's future light cone. This fact serves as a check of the correctness of our calculations.\footnote{We did not manually truncate the Green's functions in Figure~\ref{plots}. We calculated the functions numerically for all times within the plots' domains; outside of the solenoid's future light cone, the two terms cancel identically because of equation~\eqref{causality}.}

\begin{figure*}
\includegraphics[width=0.325\textwidth]{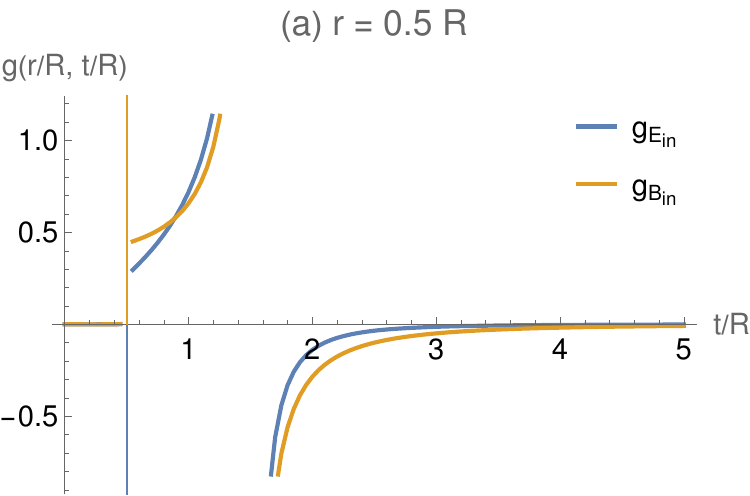}
\includegraphics[width=0.325\textwidth]{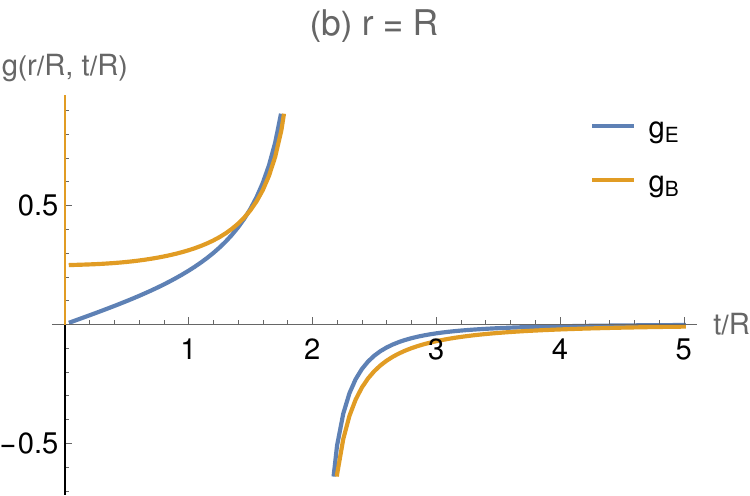}
\includegraphics[width=0.325\textwidth]{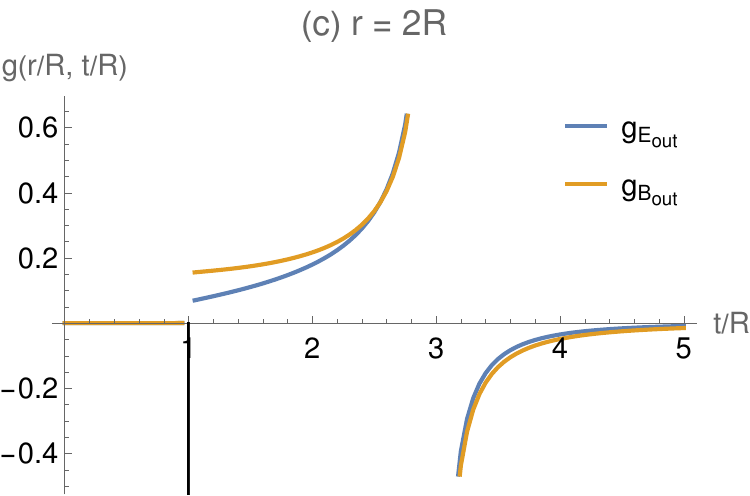}
\caption{\label{plots} Plots of the time-domain Green's functions $g_E(r, t)$ for the electric field (in blue) and $g_B(r, t)$ for the magnetic field (in orange) (a) inside the solenoid, at $r = 0.5 R$, (b) at the solenoid, at $r = R$, and (c) outside the solenoid, at $r = 2R$. The Green's functions vanish identically if $t < |R - r|$ (outside the solenoid's future light cone). Each Green's function has a $-1/(t - (R + r))$ singularity at $t = r + R$. The vertical rays at $t = |r - R|$ represent Dirac delta functions. Both Green's functions are non-differentiable at the solenoid $r = R$, but $g_E(r, t)$ is continuous there, while $g_B(r, t)$ is continuous except for the delta-function term, whose coefficient jumps discontinuously from being positive inside the solenoid to being negative outside of it. Upward-pointing rays indicate a delta function with a positive coefficient, downward-pointing rays indicate one with a negative coefficient, and black rays indicate a blue and orange ray overlaid at the same point in time. Panel (b) displays orange rays pointing both upward and downward to indicate the sign change across the solenoid. We calculated the Green's functions $g(r, t) = g_\text{SI}(r, t) + g_\text{HF}(r, t)$ from equations~\eqref{gtilde}, \eqref{ghighfreq}, \eqref{gtildeSI}, \eqref{gSI}, and \eqref{gHFrt}. We discretized the independent variable $t$ into time steps $\Delta t = 0.05 R$ and evaluated the improper integral~\eqref{gSI} numerically. The integral appears to converge, but its integrand oscillates in sign, so for the purpose of numerical stability we truncated the upper bound of integration to $\omega_\text{max} R = 2000$.}
\end{figure*}

\appendix

\section{Vanishing electromagnetic fields outside of an ideal solenoid} \label{Vanishing}

As mentioned in the main text, Ref.~\onlinecite{Abbott} proves a truly remarkable theorem about infinite ideal solenoids. This theorem is not entirely simple to state, but it is worth the effort.

\begin{theorem}
Suppose that $\kappa(t)$ is \emph{any} continuous real-valued function defined over a time interval of finite duration $T$. Then there exists an extension $K(t)$ of $\kappa(t)$ to the entire real line such that the following propositions hold:
\begin{itemize}
\item $K(t)$ is defined over all of $\mathbb{R}$.
\item $K(t) \equiv \kappa(t)$ on the original domain of $\kappa(t)$.
\item If $K(t)$ is the current density for an infinite ideal solenoid of diameter $cT$ (where $c$ is the speed of light), then the \emph{exact} electromagnetic fields vanish \emph{identically} outside the solenoid at all times.\qed
\end{itemize}
\end{theorem}

(In fact, the original function $\kappa(t)$ does not even have to be continuous; it can have a finite number of jump discontinuities.) The extended function $K(t)$ is \emph{not} periodic in general.

\section{Modified inverse Fourier transforms evaluated over positive frequencies} \label{Analytic}

Let
\[
\mathcal{F}[f(t)] := \int_{-\infty}^\infty f(t)\, e^{i \omega t}\, dt
\]
and
\[
\mathcal{F}^{-1} \big[ \tilde{f}(\omega) \big] := \frac{1}{2\pi} \int_{-\infty}^\infty \tilde{f}(\omega)\, e^{-i \omega t}\, d\omega
\]
denote the forward and inverse Fourier transforms and let $\theta(x)$ denote the Heaviside step function. Let the \emph{Hilbert transform} of a function $f(t)$ be another function $\hat{f}(t)$ or $H[f](t)$ defined by
\beq \label{Hilbert}
\hat{f}(t) = H[f](t) := \frac{1}{\pi}\, \mathrm{p.v.} \int_{-\infty}^\infty \frac{f(\tau)}{t-\tau} d\tau,
\eeq
where $\mathrm{p.v.} \int$ denotes the Cauchy principal value of an improper integral \cite{King}.\footnote{Strictly speaking, the Hilbert transform is only defined for functions in $L^p$ with $p \in (1, \infty)$ by the Titchmarsh theorem \cite{King}. But in this appendix, we will ignore questions of functional domains and assume that all functions are sufficiently well-behaved. It will sometimes be clearer for us to write $H[f(t)]$ instead of $H[f](t)$ or $\hat{f}(t)$; these expressions all refer to the same Hilbert transform.} Note that unlike with the Fourier transform, the arguments of $f(t)$ and $\hat{f}(t)$ lie in the same domain.

All the results in this Appendix follow from the follow theorem:

\begin{theorem}
If $\tilde{f}(\omega) = \mathcal{F}[f(t)]$ and $f(t) = \mathcal{F}^{-1}\big[ \tilde{f}(\omega) \big]$ are a Fourier pair of functions $\mathbb{R} \to \mathbb{C}$, then
\beq \label{AnalyticTheorem}
\mathcal{F}^{-1} \big[ 2 \theta(\omega) \tilde{f}(\omega) \big] = f(t) - i \hat{f}(t),
\eeq
where $\hat{f}(t)$ denotes the Hilbert transform $H[f](t)$.\footnote{Under the standard sign convention for a \emph{spatial} Fourier transform, where the \emph{forward} Fourier transform contains $e^{-i k x}$ and the \emph{inverse} transform contains $e^{i k x}$, the RHS becomes $f(t) + i \hat{f}(t)$.}
\end{theorem}
\begin{proof}
\begin{align*}
\mathcal{F}^{-1}\big[ 2 \theta(\omega) \tilde{f}(\omega) \big] &= \mathcal{F}^{-1} \big[ \left( 1 + \sgn(\omega) \right) \tilde{f}(\omega) \big] \\
&= f(t) + \mathcal{F}^{-1} \big[ \sgn(\omega) \tilde{f}(\omega) \big].
\end{align*}
By the convolution theorem, the final inverse Fourier transform
\begin{align*}
\mathcal{F}^{-1} \big[ \sgn(\omega) \tilde{f}(\omega) \big] &= \mathcal{F}^{-1} \left[ \sgn(\omega) \right] \ast f(t) \\
&= -\frac{i}{\pi t} \ast f(t) \\
&= -i\, H[f](t),
\end{align*}
where $\ast$ denotes convolution with respect to time.\footnote{Strictly speaking, the inverse Fourier transform of $\sgn(\omega)$ is actually the (tempered) \emph{distribution} $\mathrm{p.v.} \frac{-i}{\pi t}$ that maps a compactly supported smooth test function $f: \mathbb{R} \to \mathbb{C}$ to $\mathrm{p.v.} \int_{-\infty}^\infty \frac{-i f(t)}{\pi t}\, dt$ \cite{King}.}
\end{proof}

\begin{corollary}
If $f:\mathbb{R} \to \mathbb{C}$ is complex-valued, then
\begin{align*}
\Re \left[ \frac{1}{\pi} \int_0^\infty \tilde{f}(\omega)\, e^{-i \omega t}\, d\omega \right] &= \Re[f(t)] + \Im \big[ \hat{f}(t) \big] \\
&= \Re[f(t)] + H[\Im[f(t)]].
\end{align*}
\end{corollary}
\begin{proof}
\beq \label{1sidedproof}
\frac{1}{\pi} \int_0^\infty \tilde{f}(\omega)\, e^{-i \omega t}\, d\omega = \mathcal{F}^{-1}\big[ 2 \theta(k) \tilde{f}(\omega) \big] = f(t) - i \hat{f}(t).
\eeq
The Hilbert transform is an integral transform, so it is linear and we can apply it separately to the real and imaginary parts of $f(t)$. Its integral kernel is real-valued, so the Hilbert transform of the real part of $f(t)$ is real, and the Hilbert transform of the imaginary part of $f(t)$ is imaginary.
\end{proof}

\begin{corollary}
If $A \in \mathbb{C}$ and $t_0 \in \mathbb{R}$, then 
\beq \label{OnesidedExpIFTapp}
\Re \left[ \frac{1}{\pi} \int_0^\infty A\, e^{-i \omega (t - t_0)}\, d\omega \right] = \Re[A]\, \delta(t - t_0) + \frac{\Im[A]}{\pi (t - t_0)}.
\eeq
\end{corollary}
\begin{proof}
Let $f(t) = A\, \delta(t - t_0)$ and $\tilde{f}(\omega) = A\, e^{i \omega t_0}$. Then
\[
\hat{f}(t) = \frac{A}{\pi (t - t_0)}
\]
from definition~\eqref{Hilbert}, and
\[
\frac{1}{\pi} \int_0^\infty A\, e^{-i \omega (t - t_0)}\, d\omega = A\, \delta(t - t_0) - i \frac{A}{\pi(t-t_0)}.
\]
Equation~\eqref{OnesidedExpIFTapp} follows from decomposing $A$ into its real and imaginary parts and then taking the real part of the above equation.
\end{proof}
We use equation~\eqref{OnesidedExpIFTapp} in equation~\eqref{OnesidedExpIFTbody} of the main text.

\begin{corollary}
If $f:\mathbb{R} \to \mathbb{R}$ is real-valued, then
\[
f(t) \equiv \Re \left[ \frac{1}{\pi} \int_0^\infty \tilde{f}(\omega)\, e^{-i \omega t}\, d\omega \right].
\]
\end{corollary}

Letting $f(t)$ be the surface current $K(t)$ discussed in the main text proves eq.~\eqref{FourierK}.

The proof of the main theorem is rather formal, but there is a way to understand it more intuitively. Because we take the real part in eq.~\eqref{Komega}, the sign of $\omega$ is ambiguous; the complex constants $\tilde{K}_\omega$ and $\tilde{K}^*_{-\omega}$ correspond to the same physical (i.e. real-valued) current $K(t)$. This corresponds to the fact that the Fourier transform of a real function $f(t)$ is Hermitian (i.e. $\tilde{f}(-\omega) \equiv \tilde{f}^*(\omega)$), while the Fourier transform of a pure imaginary function $g(t)$ is anti-Hermitian (i.e. $\tilde{g}(-\omega) \equiv -\tilde{g}^*(\omega)$). Therefore, conversely only the Hermitian part of $\tilde{K}(\omega)$ will affect the physical (i.e. real) part of $K(t)$. Loosely speaking, we have a sort of ``gauge freedom'' to add any imaginary part we want to $K(t)$ (or equivalently, to add any anti-Hermitian part to $\tilde{K}(\omega)$) without changing any physical observables.

We can think of this ``gauge freedom'' as allowing us to, for each pair of opposite-sign frequencies $\pm \omega_0$, independently redistribute Fourier weights between $\tilde{K}(\omega_0)$ and $\tilde{K}^*(-\omega_0)$, since only their sum is physically meaningful. The most natural ``gauge fixing'', corresponding to the usual Fourier transform of a pure real function $K(t)$, is to distribute the amplitude (and phase) of each real mode with frequency $\omega_0$ symmetrically between $\tilde{K}(\omega_0)$ and $\tilde{K}^*(-\omega_0)$. This choice leads to a Hermitian function $\tilde{K}(\omega)$. Adding a pure imaginary component to $K(t)$ before taking the Fourier transform redistributes the Fourier weights non-symmetrically between the pairs $\tilde{K}(\omega_0)$ and $\tilde{K}^*(-\omega_0)$ by adding an anti-Hermitian part to $\tilde{K}(\omega)$.

The Hermitian ``gauge choice'' is \emph{usually} the most convenient, but in some situations other choices are more convenient. The most common alternative choice is to add a pure imaginary component to the original real signal that corresponds to placing all of the Fourier weights on the nonnegative frequencies only. A complex-valued function $f(t)$ whose Fourier transform is nonzero only for $\omega \geq 0$ (i.e. $\tilde{f}(\omega) \equiv 0$ if $\omega < 0$) is referred to as an \emph{analytic signal} \cite{Smith}.

If $f:\mathbb{R} \to \mathbb{R}$ is a \emph{real}-valued function, then its \emph{analytic representation} $f_a:\mathbb{R} \to \mathbb{C}$ is defined to be the complex-valued function $f_a(t) := f(t) - i \hat{f}(t)$. By theorem~\ref{AnalyticTheorem}, $f_a(t) = \mathcal{F}^{-1} \big[ 2 \theta(\omega) \tilde{f}(\omega) \big]$. Every analytic representation of a real-valued function is a (complex-valued) analytic signal, and vice versa:

\begin{corollary}
$f_a:\mathbb{R} \to \mathbb{C}$ is an analytic signal iff $f_a(t) = f(t) - i \hat{f}(t)$ for some $f:\mathbb{R} \to \mathbb{C}$.
\end{corollary}
\begin{proof}
For the forward direction: if $f_a$ is an analytic signal then $\tilde{f}_a(\omega) \equiv \theta(\omega) \tilde{f}_a(\omega)$. Taking the inverse Fourier transform of both sides and applying theorem~\eqref{AnalyticTheorem} (but without the factor of $2$), we find that
\[
f_a(t) \equiv -i\, \hat{f}_a(t).
\]
If we write $f_a(t) = \Re[f_a(t)] + i \Im[f_a(t)]$ and equate the imaginary parts of the equation above, then we get that $\Im[f_a(t)] = -\Re \big[ \hat{f}_a(t) \big] = -H[\Re[f_a(t)]]$ because the Hilbert transform's integral kernel is real. Let $f(t) = \Re[f_a(t)]$.

For the reverse direction: Theorem~\eqref{AnalyticTheorem} gives that $f_a(t) = \mathcal{F}^{-1}\big[2 \theta(\omega) \tilde{f}(\omega) \big]$. Taking the forward Fourier transform of both sides gives that $\tilde{f}_a(\omega) = 2 \theta(\omega) \tilde{f}(\omega)$. Therefore, $f_a(t)$ is an analytic signal.
\end{proof}

This corollary is essentially the Fourier dual of the Kramers-Kronig relations. The Kramers-Kronig relations describe the \emph{forward} Fourier transform of a causal response function that vanishes for negative \emph{time displacement}, while this corollary describes the \emph{inverse} Fourier transform of a function that vanishes for negative \emph{frequency}. For square-integrable functions $f_a$, the corollary can be extended to the statement that the analytic continuation of $f_a$ to the complex plane is analytic on the closed upper half-plane and approaches zero on the upper half-plane as $|z| \to \infty$; this is another version of Titchmarsh's theorem \cite{Titchmarsh}.

We can now understand what is going on behind the scenes in eq.~\eqref{FourierK} in the main text. We can think of the ``one-sided inverse Fourier transform'' 
\[
\frac{1}{\pi} \int \limits_0^\infty \tilde{K}(\omega)\, e^{-i \omega t}\, d\omega
\]
inside the brackets in eq.~\eqref{FourierK} as $\mathcal{F}^{-1} \big[ 2 \theta(\omega) \tilde{f}(\omega) \big]$. By theorem~\eqref{AnalyticTheorem}, this is not the real-valued physical function $K(t)$, but instead the analytic representation $K_a(t) = K(t) - i \hat{K}(t)$ with only positive-frequency components. Taking the real part drops the imaginary Hilbert transform term, restores the negative-frequency components, and leaves the physical function $K(t)$.

We can prove (at a physicist's level of rigor) a ``one-sided'' variant of the convolution theorem:
\begin{corollary}
If $\big( f(t), \tilde{f}(\omega) \big)$ and $\left( g(t), \tilde{g}(\omega) \right)$ are each Fourier pairs of functions $\mathbb{R} \to \mathbb{C}$, then
\begin{align*}
\frac{1}{\pi} \int_0^\infty \tilde{f}(\omega)\, \tilde{g}(\omega)\, e^{-i \omega t}\, d\omega &= \big( f(t) - i\, \hat{f}(t) \big) \ast g(t) \\
&= f(t) \ast \left( g(t) - i\, \hat{g}(t) \right). \nn
\end{align*}
\end{corollary}
\begin{proof}
\begin{align}
\frac{1}{\pi} \int_0^\infty \tilde{f}(\omega)\, \tilde{g}(\omega)\, e^{-i \omega t}\, d\omega &= \mathcal{F}^{-1} \big[ \tilde{f}(\omega)\, 2 \theta(\omega)\, \tilde{g}(\omega) \big] \nn \\
&= \mathcal{F}^{-1} \big[ \tilde{f}(\omega)\, 2 \theta(\omega) \big] \ast g(t) \label{convproof} \\
&= f(t) \ast \mathcal{F}^{-1} \big[ 2 \theta(\omega)\, \tilde{g}(\omega) \big]. \nn
\end{align}
\end{proof}

More generally, we have that
\beq \label{general}
\frac{1}{\pi} \int_0^\infty \prod_{k=1}^n \tilde{f}_k(\omega)\, e^{-i \omega t}\, d\omega
\eeq
equals $1/2^{m-1}$ times the $n$-fold convolution of the $\{f_k(t)\}$, $k = 1, ... n$, but with any $m$ of the convolved functions $f_k(t)$ replaced by $f_k(t) - i\, \hat{f}_k(t)$, where $m$ is any natural number $m = 1, \dots, n$. This is because $2 \theta(\omega) = 2 \theta(\omega)^m = \frac{1}{2^{m-1}} (2 \theta(\omega))^m$, and we can multiply each factor $\tilde{f}_k(\omega)$ in the integrand of the integral~\eqref{general} by a different factor $2 \theta(\omega)$.

Finally, we have
\begin{corollary} \label{1sidedconv}
If $\left( g(t), \tilde{g}(\omega) \right)$ is a Fourier pair and $g:\mathbb{R} \to \mathbb{R}$ is real-valued, then
\begin{multline} \label{reonesidedconv}
\mathrm{Re} \left[ \frac{1}{\pi} \int_0^\infty \tilde{f}(\omega)\, \tilde{g}(\omega)\, e^{-i \omega t}\, d\omega \right] \\= \mathrm{Re} \left[ \frac{1}{\pi} \int_0^\infty \tilde{f}(\omega)\, e^{-i \omega t}\, d\omega \right] \ast g(t).
\end{multline}
\end{corollary}
\noindent The proof follows from equations~\eqref{convproof} and \eqref{1sidedproof} and from the fact that $g(t)$ is real.

Letting $g(t)$ be the surface current $K(t)$ discussed in the main text and using equation~\eqref{Einrt} proves eqs.~\eqref{convs}.

\section{Asymptotics of the frequency-domain Green's functions for the electromagnetic fields} \label{Asymptotics}

The Bessel functions of the first kind $J_0(x)$ and $J_1(x)$ and the Hankel functions of the first kind $H^{(1)}_0(x)$ and $H^{(1)}_1(x)$ have the asymptotic forms
\begin{align*}
J_0(x) &\sim 1, \\
J_1(x) &\sim \frac{1}{2} x \\
H^{(1)}_0(x) &\sim \frac{2 i}{\pi} \ln x \\
H^{(1)}_1(x) &\sim - \frac{2 i}{\pi x}
\end{align*}
as $x \to 0^+$ and
\begin{subequations} \label{Jas}
\begin{align}
J_0(x) &\sim \sqrt{\frac{2}{\pi x}} \cos \left( x - \frac{1}{4} \pi \right) \\
J_1(x) &\sim \sqrt{\frac{2}{\pi x}} \cos \left( x - \frac{3}{4} \pi \right) \end{align}
\end{subequations}
\vspace{-\baselineskip}
\begin{subequations} \label{Has}
\begin{align}
H^{(1)}_0(x) &\sim \sqrt{\frac{2}{\pi x}} e^{i \left( x - \frac{1}{4} \pi \right)} \hspace{44pt} \\
H^{(1)}_1(x) &\sim \sqrt{\frac{2}{\pi x}} e^{i \left( x - \frac{3}{4} \pi \right)}
\end{align}
\end{subequations}
as $x \to +\infty$ \cite{Jackson}.\footnote{Strictly speaking, the asymptotic relations~\eqref{Jas} for the Bessel functions $J_n(x)$ are not quite correct in the standard sense, because the zeros of the Bessel functions come arbitrary close to -- but never exactly equal to -- the zeros of the cosine approximation at $\left( m + \frac{1}{2} \left(n - \frac{1}{2} \right) \right) \pi$, $m \in \mathbb{Z}$. Those relations should be interpreted as saying that the ratio of the two sides comes arbitrarily close to 1 as $x \to +\infty$, except on an infinite sequence of real intervals where the ratio diverges -- but those intervals (where the ratio deviates from $1$ by more than $\epsilon$ for any fixed $\epsilon$) become arbitrarily short as $x \to +\infty$. The asymptotic relations~\eqref{Has} for the Hankel functions $H^{(1)}_n(x)$ are correct in the standard sense.}

Therefore, the frequency-domain Green's functions~\eqref{gtilde} have the asymptotic forms
\begin{align}
\tilde{g}_{E_\text{in}}(r, \omega) &\sim \frac{i r}{2} \omega \label{lowfreqgs} \\
\tilde{g}_{E_\text{out}}(r, \omega) &\sim \frac{i R^2}{2r} \omega \nn \\
\tilde{g}_{B_\text{in}}(r, \omega) &\sim 1 \nn \\
\tilde{g}_{B_\text{out}}(r, \omega) &\sim -\frac{1}{2} R^2 \omega^2 \ln(\omega r) \nn
\end{align}
as $\omega \to 0^+$ and
\begin{align}
\tilde{g}_{E_\text{in}}(r, \omega) &\sim \sqrt{\frac{R}{r}} \cos \left( \omega r - \frac{3}{4} \pi \right) e^{i \left( \omega R + \frac{1}{4} \pi \right)} \label{highfreqgs} \\
\tilde{g}_{B_\text{in}}(r, \omega) &\sim \begin{cases}
\sqrt{\frac{R}{r}} \cos \left( \omega r - \frac{1}{4} \pi \right) e^{i \left( \omega R - \frac{1}{4} \pi \right)} \qquad &\text{if } r > 0 \\
\sqrt{\frac{\pi \omega R}{2}} e^{i \left( \omega R - \frac{1}{4} \pi \right)} \qquad &\text{if } r = 0
\end{cases} \nn \\
\tilde{g}_{E_\text{out}}(r, \omega) &\sim \tilde{g}_{B_\text{out}}(r, \omega) \sim \sqrt{\frac{R}{r}} \cos \left( \omega R - \frac{3}{4} \pi \right) e^{i \left( \omega r + \frac{1}{4} \pi \right)} \nn
\end{align}
as $\omega \to +\infty$.\footnote{The high-frequency asymptotic relations~\eqref{highfreqgs} as $\omega \to +\infty$ hold in the same sense as the asymptotic relations~\eqref{Jas} discussed in the previous footnote. Recall that $\tilde{g}_{E_\text{in}}(r = 0, \omega) \equiv 0$.} By expanding out the cosines as sums of complex exponentials, we can rewrite the high-frequency asymptotics~\eqref{highfreqgs} as equation~\eqref{ghighfreq} in the main text. This equation shows that $\tilde{g}_{E_\text{in}}(r, \omega) \sim \tilde{g}_{E_\text{out}}(r, \omega) \sim \tilde{g}_{B_\text{out}}(r, \omega)$ as $\omega \to +\infty$, which is not immediately obvious from equations~\eqref{highfreqgs}.

\bibliography{Solenoid}
\bibliographystyle{apsrev4-2}

\end{document}